\documentclass[a4paper]{article}
\usepackage[english]{babel}
\usepackage{amscd}
\usepackage{amsmath}
\usepackage{amsthm}
\usepackage{amssymb}
\usepackage{graphicx}

 \usepackage{geometry}

\usepackage{wrapfig}
\usepackage[section]{placeins}
\usepackage{epigraph}
\usepackage{longtable}

\usepackage[usenames]{color}
\usepackage{colortbl}

\usepackage{titlesec}

\titleformat{\subsection}[runin]{\normalfont\bfseries}{\thesubsection}{1em}{}

\usepackage{mathtools}
\newcommand{\defeq}{\vcentcolon=}

\usepackage[numbers, sort&compress]{natbib}

\theoremstyle{plain}
    \newtheorem{lemma}{Lemma}[section]
    \newtheorem{theorem}{Theorem}
    
    \newtheorem{statement}{Proposition}[section]

\theoremstyle{definition}
    \newtheorem{example}{Example}[section]
      \newtheorem{remark}{Remark}[section]
    \newtheorem{definition}{Definition}

\newcommand{\Ker}[1]{\mathrm{Ker} \, #1}

\newcommand{\rank}[1]{\mathrm{rank} \, #1}

\newcommand{\corank}[1]{\mathrm{corank} \, #1}

\newcommand{\diff}[1]{\mathrm{d}  #1}

\newcommand{\difft}[1]{ \frac{\mathrm d }{\mathrm d #1} }

\newcommand{\R}{\mathbb{R}}
\newcommand{\Complex}{\mathbb{C}}

\newcommand{\T}{\mathrm{T}}
\newcommand{\Cont}{\mathrm{C}}

\newcommand{\CP}{\overline{\mathbb{C}}}
\newcommand{\RP}{\overline{\mathbb{R}}}
\newcommand{\Hom}{\mathrm{H}}

\newcommand{\Tr}[1]{\mathrm{Tr} \, #1}

\newcommand{\eps}{\varepsilon}

\newcommand{\LieBracket}{ [\, , ] }
\newcommand{\PoissonBracket}{ \{ \, , \} }
\newcommand{\g}{\mathfrak{g}}
\newcommand{\h}{\mathfrak{Ker}}

\newcommand{\so}{\mathfrak{so}}
\newcommand{\SO}{\mathrm{SO}}

\newcommand{\ad}{\mathrm{ad}}
\newcommand{\zenter}{\mathcal{Z}}

\newcommand{\M}{{M}}

	\newcounter{pd}
\title{Stability of relative equilibria of multidimensional rigid body}
\author{Anton Izosimov\footnote{Lomonosov Moscow State University and National Research University Higher School of Economics. e-mail: izosimov@mech.math.msu.su}}
\date{}
\begin{document}
\maketitle
\abstract{It is a classical result of Euler that the rotation of a torque-free three-dimensional rigid body about the short or the long axis is stable, whereas the rotation about the middle axis is unstable. This result is generalized to the case of a multidimensional body.}
\section{Introduction}
\subsection{Three-dimensional free rigid body}
The Euler problem in rigid body dynamics is one of the following equivalent problems.
\begin{enumerate}
\item The motion of a rigid body fixed at the center of mass under no external forces.
\item The motion of a rigid body which is free to move in space under no external forces.
\end{enumerate}
The second problem is reduced to the first one by passing to the coordinate system related to the center of mass. In both cases we can add a constant gravity field because the resulting torque of the gravity force with respect to the center of mass vanishes.
\par
Let us consider the problem of motion of a rigid body fixed at the center of mass acted on by no external forces. Then, as was observed by Euler, the evolution equations for the angular velocity do not involve the position coordinates of the body. Euler's equations have the form
\begin{align}\label{EulerForm}
\begin{aligned}
\, I_1\,\dot \omega_1 &= (I_2 - I_3)\,\omega_2\omega_3 ,\\
\, I_2\,\dot \omega_2&= (I_3 - I_1)\,\omega_3\omega_1,\\
\, I_3\,\dot \omega_3 &= (I_1 - I_2)\,\omega_1\omega_2
\end{aligned}
\end{align}
where $I_1, I_2, I_3$ are the principal moments of inertia, and  $\omega = (\omega_1, \omega_2, \omega_3)$ is the angular velocity vector written in principal axes. In terms of modern geometric mechanics, this system is obtained from the ``rigid body fixed at the center of mass'' problem by reduction with respect to the $\SO(3)$ action. 
\begin{remark}
Recall that the inertia tensor of a rigid body is a positive-definite quadratic form $\mathbb I$ which characterizes the distribution of mass in the body. The eigenvalues of $\mathbb I$ are called principal moments of inertia, and its eigenvectors are called principal axes. For a uniform box-shaped body, principal axes coincide with the axes of symmetry, and principal moments are inverse proportional to the lengths of these axes. See \cite{Arnold} for details.
\end{remark}
If we assume that the body is asymmetric, i.e. if $I_1, I_2,$ and $I_3$ are pairwise distinct, then the right-hand sides of the equations \eqref{EulerForm} vanish simultaneously if and only if the angular velocity vector is collinear to one of the three principal axes. Thus the fixed point set of the system \eqref{EulerForm} consists of three mutually orthogonal straight lines. These fixed points are stationary, or permanent, rotations of the body, i.e. such motions that the axis of rotation is time-independent. 
Stationary rotations are also called relative equilibria. \par

As was shown by Euler, stationary rotations about different principal axes have different dynamical features. Rotation about the axis of greatest moment of inertia or axis of least moment of inertia is stable, whereas rotation about the intermediate axis is unstable. 
This can be demonstrated by trying to spin a book about one of its symmetry axes. While the book spins fairly well about the longest and the shortest axis, spinning about the intermediate axis causes the book to ``tumble'', periodically reversing the direction of rotation.
%
%
\par The aim of the present paper is to establish a multidimensional generalization of this result. The problem was studied by a number of authors \cite{Oshemkov, Marshall, Spiegler, Casu2, Casu, Ratiu}, however the general answer has only been obtained in dimension four. As the dimension grows, the problem becomes too complicated from the computational point of view when being approached by direct methods. In the present paper, the problem is solved in arbitrary dimension by means of algebraic technique related to compatible Poisson brackets and Lie algebras.
\subsection{Multidimensional rigid body}
The possibility to generalize the free rigid body equations to the $n$-dimensional case was already mentioned by Frahm \cite{Frahm} and Weyl \cite{Weyl}.
Arnold \cite{Arnold} observed that, after the standard identification of $\R^3$ with the space of skew-symmetric $3\times 3$ matrices $\so(3)$, the equations \eqref{EulerForm} can be rewritten in the form
 \begin{align}\label{eae0}
\begin{cases}
 \dot \M= [\M, \Omega],\\
\M = \Omega J + J \Omega
\end{cases}
 \end{align}
 where
$\M \in \so(3)$ is the angular momentum,
 $\Omega \in \so(3)$ is the angular velocity, and $J = \mathrm{diag}(J_1, J_2, J_3)$ is a constant positive-definite diagonal matrix such that $I_1 = J_2 + J_3, I_2 = J_1 + J_3$, and $I_3 = J_1 + J_2$.\par
The multidimensional generalization of the equations \eqref{eae0} is straightforward: we just replace $3 \times 3$ matrices by $n \times n$ matrices. A somewhat better approach is to generalize not the equations but the problem. Consider an $n$-dimensional rigid body fixed at the center of mass acted on by no external forces. Fix a space frame and a body frame both centered at the center of mass of the body. Let $X(t) \in \SO(n)$ be the position of the body frame with respect to the space frame. Define $\Omega = X^{-1}\dot X$. This matrix is skew-symmetric and is called the angular velocity matrix. Define also a symmetric matrix $J$ by
$$
J_{ij} = \int x_ix_j \diff \mu
$$
where the coordinates $x_i$ are related to the body frame, and $\diff \mu$ is the density of the mass distribution. Then it can be proved that the evolution of the angular velocity matrix $\Omega$ is governed by the equations \eqref{eae0}. Note that the equations  \eqref{eae0} are equivalent to the conservation of the angular momentum in the space frame:
$$
	\difft{t}(XMX^{-1})= 0.
$$ See \cite{Arnold, Manakov2} for details.


\begin{remark}{Following \cite{FK}, we suggest that $J$ is called the mass tensor. The mass tensor should not be confused with the inertia tensor. The inertia tensor is the map $\mathbb I \colon \so(n) \to \so(n)$ which is given by
$
\mathbb I(\Omega) = J\Omega + \Omega J.
$ In three dimensions $\so(3)$ may be identified with $\R^3$, which may lead to a confusion between $\mathbb I$ and $J$. In higher dimensions these two operators act on different spaces.} \end{remark}
\subsection{Multidimensional rigid body as a completely integrable system}
Arnold showed that the system (\ref{eae0}) is Hamiltonian with respect to the Lie-Poisson bracket on the dual of the Lie algebra $\so(n)$, and therefore the invariants of the coadjoint representation are first integrals of the system. These first integrals are trivial in the sense that they are Casimir functions of the Lie-Poisson bracket and do not correspond to symmetries. Later, Mischenko \cite{Mischenko} found a family of non-trivial quadratic first integrals. They were shown to be involution with respect to the Lie-Poisson bracket by Dikii \cite{Dikii}. Dikii also observed that in the four-dimensional case Mischenko's first integrals are sufficient for complete Liouville integrability.
In his famous paper \cite{Manakov}, Manakov showed that the equations \eqref{eae0} can be rewritten in the form
\begin{align*}
\difft{t}(\M+\lambda J^2) = [\M+\lambda J^2, \Omega + \lambda J],
\end{align*}

 which implies that the functions
 $
f_{\lambda, k}= \Tr (\M+\lambda J^2)^k
 $
 are first integrals\footnote{Manakov also showed that the system (\ref{eae0}) can be embedded into a large class of integrable systems which are now called Manakov tops. Note that all the results of the present paper remain true for all generic Manakov tops.}.
Later it was proved by Fomenko, Mischenko \cite{MF} and Ratiu \cite{Manakov2} that these first integrals Poisson-commute and are sufficient for complete Liouville integrability.\par
Note that we will not be using the integrals of the system in their explicit form: they are complicated polynomials uneasy to deal with. Instead of considering the integrals, we will make use of the bi-Hamiltonian structure of the system which encodes all the information about them. The bi-Hamiltonian structure of  \eqref{eae0}  was discovered by Bolsinov \cite{Bolsinov, Bolsinov2}\footnote{The same bi-Hamiltonian structure was later rediscovered by  Morosi and Pizzocchero \cite{Manakov3}.}.

\subsection{Stability for the multidimensional rigid body}
Below we discuss what is known about stability of stationary rotations in the multidimensional case. A more detailed comparison of previously known results with the results of the present paper can be found in Section \ref{exSect}.\par
The first topological description of the four-dimensional rigid body problem was obtained by Oshemkov \cite{Oshemkov} who constructed bifurcation diagrams of the moment map.  As it is clear now, these diagrams can be used, in principle, to study stability of stationary rotations\footnote{See \cite{ts} where the relation between topology and stability in integrable systems is discussed. Also note that the topological approach to integrable systems was first proposed in the classical work of Smale \cite{Smale}.}.\par
The solution of the stability problem in dimension four was obtained by Feher and Marshall \cite{Marshall}, and later by another approach by Birtea and Ca\c{s}u \cite{Casu2}, Birtea, Ca\c{s}u, Ratiu, and Turhan \cite{Ratiu}. 
In the present paper these results receive a geometric interpretation (see Example \ref{4d}).\par
There was also an attempt to solve the stability problem in five dimensions, however only partial results are available, see Ca\c{s}u \cite{Casu}. \par

The multidimensional situation was studied in the thesis of
 Spiegler \cite{Spiegler}. He gave a sufficient condition for a stationary rotation to be stable in arbitrary even dimension. However, as it follows from the results of the present paper, this condition is far from necessary and sufficient (see Example \ref{spieler}). \par 
 In the present paper, the stability problem is solved almost completely in arbitrary dimension by means of the bi-Hamiltonian approach. Bi-Hamiltonian approach for studying topology and stability in integrable systems was suggested and developed in \cite{biham, JGP, SBS}. \par
 We also refer the reader to the author's preprint \cite{AS} where the stability problem for the multidimensional rigid body is studied by means of algebraic geometry.
\subsection{Structure of the paper}
The paper is organized as follows. All main results are contained in Section \ref{StabSection}. Section  \ref{StabSection1} is devoted to the classification of stationary rotations. In Section  \ref{StabSection2}, the notion of a parabolic diagram is defined and the main stability theorem is formulated. Section  \ref{exSect} contains some examples and compares the results of the paper to previously known results.
 In Section \ref{biHam}, the machinery which allows to prove stability in bi-Hamiltonian systems is presented. In Section \ref{biHamMRB}, the bi-Hamiltonian structure of the multidimensional rigid body is introduced.
 The proof of the main theorem is in Section \ref{instabSect} and Section \ref{stabSect}.
 Finally, the appendix contains the explicit classification of Lie algebras $\g_\lambda$ which arise in the bi-Hamiltonian geometry of the multidimensional rigid body.

  \section{Main results}\label{StabSection}
\subsection{Stationary rotations}\label{StabSection1}
We study the equations
\begin{align}\label{eae}
\begin{cases}
 \dot \M= [\M, \Omega],\\
\M = \Omega J + J \Omega
\end{cases}
 \end{align}
 where
$\M \in \so(n)$ is called the angular momentum matrix\footnote{To be precise, the angular momentum $M$ belongs to the dual space $\so(n)^*$. In what follows, we identify $\so(n)$ and $\so(n)^*$ by means of the Killing form $\Tr XY$.}, 
 $\Omega \in \so(n)$ is called the angular velocity matrix, and $J$ is a positive symmetric matrix called the mass tensor. Without loss of generality we may assume that $J$ is diagonal.\par
Before describing stationary rotations, consider how an $n$-dimensional body may rotate. At each moment of time, the angular velocity matrix $\Omega$ may be brought to a canonical form
\begin{align}\label{omegaForm}
\Omega = \left(\begin{array}{ccccccc}0 & \omega_1 &  &  &  &  &    \\-\omega_1 & 0 &  &  &  &  &    \\ &  & \ddots &  &  &  &    \\ &  &  & 0 & \omega_m &  &    \\ &  &  & -\omega_m & 0 &  &    \\ &  &  &  &  & 0 &    \\ &  &  &  &  &  & \ddots  \end{array}\right)
\end{align} 
by an orthogonal transformation.
In other words, $\R^n$ is decomposed into a sum of $m$ pairwise orthogonal two-dimensional planes $\Pi_1, \dots, \Pi_m$ and a space $\Pi_0$ of dimension $n-2m$ orthogonal to all these planes:
\begin{align}\label{decomp}
\R^n = \left(\bigoplus_{i=1}^{m} \Pi_i \right)\oplus \Pi_0.
\end{align}
 There is an independent rotation in each of the planes $\Pi_1, \dots, \Pi_m$, while $\Pi_0 = \Ker \Omega$ is immovable.

\begin{definition}
The eigenvectors of $J$ are called \textit{principal axes of inertia}.
\end{definition}
A classical three-dimensional result states that a rotation is stationary if and only if it is a rotation about a principal axis of inertia. In the multidimensional case, this is not always so. Stationary rotations of a multidimensional rigid body are described in \cite{NRE}. 

 \begin{statement}
Consider the system (\ref{eae}). Suppose that $J$ has pairwise distinct eigenvalues. Then $\M$ is an equilibrium point of the system if and only if there exists an orthonormal basis such that $J$ is diagonal, and $\Omega$ is block-diagonal of the following form
\begin{align*}
 \Omega = \left(\begin{array}{cccccc}\omega_1\Omega_1 &  & & & \\ & \ddots & & & \\ &  & \omega_{k}\Omega_k & & \\  & & & 0 & \\ & & & & \ddots  \end{array}\right),
\end{align*}
 where $\Omega_{i} \in \so(2m_i) \cap \SO(2m_i)$ for some $m_i > 0$, and $\omega_i$'s are distinct positive real numbers.
 \end{statement}

 \begin{definition}\label{regDef}
 	A stationary rotation $\M$ is \textit{regular} if there exists an orthonormal basis such that $J$ is diagonal, and $\Omega$ is of the form \eqref{omegaForm}.
Otherwise, $\M$ is \textit{exotic}. \par
 \end{definition}
 In other words, a rotation is regular if all the planes $\Pi_i$ entering (\ref{decomp}) are spanned by principal axes of inertia.
 Note that if all non-zero eigenvalues of $\Omega$ are distinct, then the rotation is automatically regular.
In the three-dimensional case, all stationary rotations are regular.\par
In the present paper, only regular stationary rotations are considered.
\begin{remark}
It was asserted in the announcement \cite{Vestnik} as well as in the earlier version of this paper that all exotic stationary rotations are unstable, however some technical details in the proof are still to be completed\footnote{
Namely, we can prove that all exotic equilibria are unstable provided that (\ref{eae}) is known to be a non-resonant system, which means that trajectories of the system are dense on almost all Liouville tori. Moreover, we can prove that the non-resonant condition is satisfied for a certain open subset of the phase space. Since the system is analytic, this should imply that it is non-resonant everywhere \cite{intsys}, however the proof of this latter assertion is unknown to the author.}. The proof will published elsewhere.
Also note that exotic stationary rotations arise as relative equilibria of the $n$ point masses problem \cite{Chenciner}.
\end{remark}
The main problem of the paper is to study regular stationary rotations for Lyapunov stability.  It will be assumed that the body is \textit{asymmetric}, i.e. all eigenvalues of $J$ are pairwise distinct.

\subsection{Parabolic diagrams and stability}\label{StabSection2}
Consider a regular stationary rotation. Then there exists an orthonormal basis such that $J$ is diagonal, and $\Omega$ is given by (\ref{omegaForm}).
In other words, there exists a decomposition (\ref{decomp}) in which all planes $\Pi_i$ are spanned by principal axes of inertia.\par
Define the notion of the parabolic diagram of a regular stationary rotation. 
\begin{enumerate} \item Draw a coordinate plane. 
\item  For each $2$-plane $\Pi_i, i > 0$, draw the parabola given by $y = \chi_i(x)$ where
\begin{align}\label{chiFormula}
\chi_{i}(x) = \frac{(x - \lambda_{i1}^{2})(x - \lambda_{i2}^{2})}{\omega_{i}^{2}(\lambda_{i1} + \lambda_{i2})^{2}},
\end{align}
 $\omega_{i}$ is the frequency of rotation in the plane $\Pi_i$,  and $ \lambda_{i1},  \lambda_{i2}$ are the eigenvalues of $J$ corresponding to the eigenvectors $e_{i1}, e_{i2} \in \Pi_i$. 
\item For each immovable principal axis $e_0 \in \Pi_0$, draw a vertical straight line through $\lambda^{2}$ where $\lambda$ is the eigenvalue of $J$ corresponding to the eigenvector $e_0$. 
\end{enumerate}
As a result, there is either a parabola or a vertical straight line passing through the square of each eigenvalue of $J$.
\begin{definition}
The obtained picture is the \textit{parabolic diagram} of a regular stationary rotation. 
\end{definition}
Parabolic diagrams for the three-dimensional rigid body are depicted in Figure \ref{3dpd}. See Section \ref{exSect} for more examples.
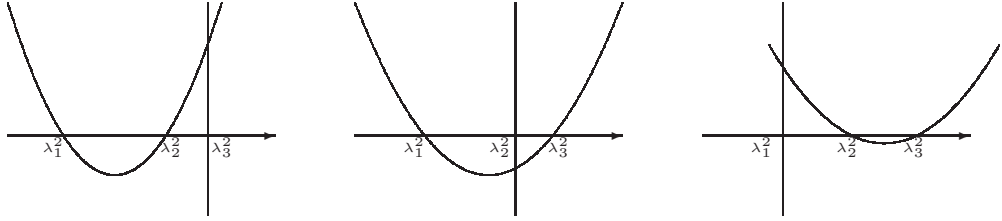
\begin{figure}[t]

{\begin{picture}(500,60)

\put(30,0){
\qbezier(0,80)(40,-50)(80,80)
\put(0,30){\vector(1,0){100}}
\put(75,0){\line(0,1){80}}
\put(13,24){\tiny{$\lambda_1^2$}}
\put(57,24){\tiny{$\lambda_2^2$}}
\put(76,24){\tiny{$\lambda_3^2$}}
}
\put(160,0){
\qbezier(0,80)(50,-50)(100,80)
\put(0,30){\vector(1,0){100}}
\put(60,0){\line(0,1){80}}
\put(18,24){\tiny{$\lambda_1^2$}}
\put(50,24){\tiny{$\lambda_2^2$}}
\put(72,24){\tiny{$\lambda_3^2$}}
}
\put(290,0){
\qbezier(25,64)(68,-10)(111,64)
\put(0,30){\vector(1,0){100}}
\put(30,0){\line(0,1){80}}
\put(18,24){\tiny{$\lambda_1^2$}}
\put(50,24){\tiny{$\lambda_2^2$}}
\put(75,24){\tiny{$\lambda_3^2$}}
}

\end{picture}}
\caption{Parabolic diagrams for the three-dimensional rigid body. Rotations around the long, middle, and short axis of inertia respectively.}\label{3dpd}
\end{figure}\par

\newpage
 The following theorem is the main result of the paper.
 
  \begin{theorem}\label{stabThm}
 Consider a regular stationary rotation of a multidimensional rigid body. 
 \begin{enumerate}
 \item Assume that
 \begin{enumerate}
  \item all intersections in the associated parabolic diagram are either real and belong to the upper half-plane, or infinite;
  \item the are no tangency points in the parabolic diagram.
  \end{enumerate}
Then the rotation is stable.
\item
Assume that
 there is at least one intersection in the parabolic diagram which is either complex or belongs to the lower half-plane. Then the rotation is unstable.
\end{enumerate}
 \end{theorem}
 \begin{remark}
 When speaking about real, complex, or infinite intersections, the parabolic diagram is considered as a curve in $\Complex \mathrm{P}^2$.
 \end{remark}
 \begin{remark}
 	In \cite{JGP}, a weaker version of the first statement of Theorem \ref{stabThm} was announced. It included an additional requirement that $\dim \Pi_0 \leq 2$. In the present paper, the technique of \cite{JGP} is extended, so that the mentioned requirement could be omitted: it seems to be quite natural to consider rotations with a large number of fixed axes.
 \end{remark}
 \begin{remark}
 	Note that this theorem solves the stability problem for an open dense subset of regular stationary rotations.

 \end{remark}

\begin{remark}
It is proved in the preprint \cite{AS} that condition 1b) of Theorem \ref{stabThm} can be omitted, so a regular stationary rotation is stable if and only if all intersections in the associated parabolic diagram are either real and belong to the upper half-plane, or infinite. The proof uses methods from algebraic geometry.
\end{remark}

\subsection{Examples}\label{exSect}

\begin{figure}[p]

{\begin{picture}(500,70)

\put(30,-10){
\qbezier(24,80)(47,-50)(70,80)
\qbezier(20,75)(100,-20)(180,75)
\put(0,30){\vector(1,0){190}}
\put(28,24){\tiny{$\lambda_1^2$}}
\put(57,24){\tiny{$\lambda_2^2$}}
\put(75,24){\tiny{$\lambda_3^2$}}
\put(112,24){\tiny{$\lambda_4^2$}}
}

\put(250,-10){
\qbezier(24,80)(47,-50)(70,80)
\qbezier(64,80)(87,-50)(110,80)

\put(0,30){\vector(1,0){140}}
\put(28,24){\tiny{$\lambda_1^2$}}
\put(57,24){\tiny{$\lambda_2^2$}}
\put(68,24){\tiny{$\lambda_3^2$}}
\put(97,24){\tiny{$\lambda_4^2$}}
}

\end{picture}}
\caption{Parabolic diagrams for the four-dimensional rigid body. Stable rotation.}\label{4dpd1}
\end{figure}
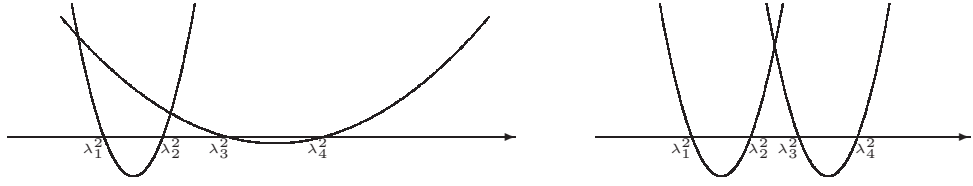
\begin{figure}[p]

{\begin{picture}(500,70)

\put(120,-10){
\qbezier(15,80)(71,-50)(127,80)
\qbezier(56,80)(91,-50)(126,80)

\put(0,30){\vector(1,0){140}}
\put(37,24){\tiny{$\lambda_1^2$}}
\put(67,24){\tiny{$\lambda_2^2$}}
\put(94,24){\tiny{$\lambda_3^2$}}
\put(107,24){\tiny{$\lambda_4^2$}}
}

\end{picture}}
\caption{Parabolic diagram for the four-dimensional rigid body. Unstable rotation.}\label{4dpd2}
\end{figure}
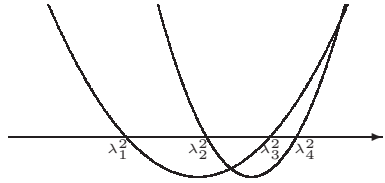
\begin{figure}[p]

{\begin{picture}(500,70)

\put(30,-0){
\qbezier(10,80)(75,-60)(140,80)
\qbezier(50,80)(80,-80)(110,80)

\put(0,30){\vector(1,0){140}}
\put(33,24){\tiny{$\lambda_1^2$}}
\put(54,24){\tiny{$\lambda_2^2$}}
\put(88,24){\tiny{$\lambda_3^2$}}
\put(107,24){\tiny{$\lambda_4^2$}}
}
\put(240,-0){
\qbezier(10,80)(75,-60)(140,80)
\qbezier(50,80)(80,-50)(110,80)

\put(0,30){\vector(1,0){140}}
\put(33,24){\tiny{$\lambda_1^2$}}
\put(56,24){\tiny{$\lambda_2^2$}}
\put(93,24){\tiny{$\lambda_3^2$}}
\put(107,24){\tiny{$\lambda_4^2$}}
}

\end{picture}}
\caption{Parabolic diagrams for the four-dimensional rigid body. Unstable rotation.}\label{4dpd3}
\end{figure}
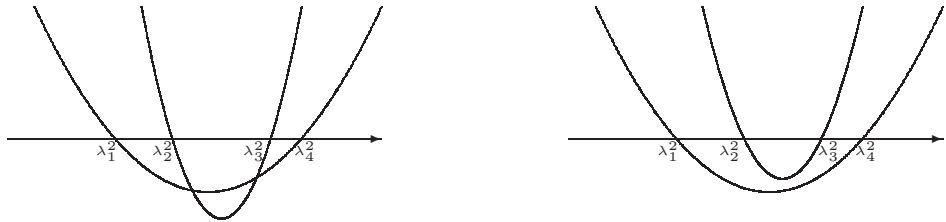
\begin{figure}[p]

{\begin{picture}(500,70)

\put(120,5){
\qbezier(30,80)(75,-100)(120,80)
\qbezier(20,80)(80,-30)(140,80)

\put(0,30){\vector(1,0){140}}
\put(37,24){\tiny{$\lambda_1^2$}}
\put(54,24){\tiny{$\lambda_2^2$}}
\put(93,24){\tiny{$\lambda_3^2$}}
\put(105,24){\tiny{$\lambda_4^2$}}
}

\end{picture}}
\caption{Parabolic diagram for the four-dimensional rigid body. Stable rotation.}\label{4dpd4}
\end{figure}
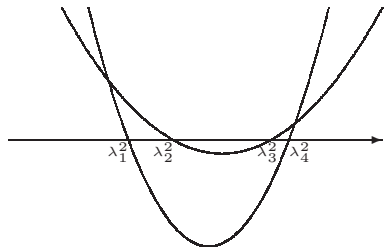
\begin{figure}[p]
{\begin{picture}(500,80)

\put(120,5){
\qbezier(30,80)(75,-100)(120,80)
\qbezier(46,80)(83,-60)(120,80)
\put(0,30){\vector(1,0){140}}
\put(37,24){\tiny{$\lambda_1^2$}}
\put(55,24){\tiny{$\lambda_2^2$}}
\put(90,24){\tiny{$\lambda_3^2$}}
\put(105,24){\tiny{$\lambda_4^2$}}
}

\end{picture}}
\caption{Parabolic diagram for the four-dimensional rigid body. Hamiltonian Hopf bifurcation.}\label{4dpd5}
\end{figure}
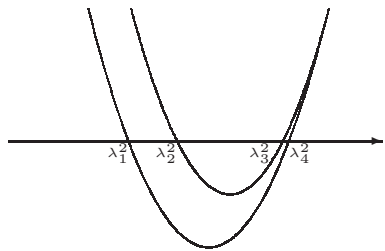
\begin{example}[Three-dimensional rigid body] Parabolic diagrams for the three-dimensional rigid body are depicted in Figure \ref{3dpd}. The classical results on stability are immediately recovered.
\end{example}\label{4d}
\begin{example}[Four-dimensional rigid body] Let $e_1, e_2,e_3,e_4$ be the principal axes sorted in order of increasing eigenvalues of $J$. There are three possibilities.
\begin{enumerate}
	\item $\Pi_1 = \langle e_1, e_2 \rangle,\Pi_2 = \langle e_3, e_4 \rangle$. The rotation is stable (see Figure \ref{4dpd1}; note that there is an intersection at infinity in the second diagram).
	\item $\Pi_1 = \langle e_1, e_3 \rangle,\Pi_2 = \langle e_2, e_4 \rangle$. The rotation is unstable (see Figure \ref{4dpd2}).
		\item $\Pi_1 = \langle e_1, e_4 \rangle,\Pi_2 = \langle e_2, e_3 \rangle$. In this case, stability depends on the ratio of angular velocities. If $\omega_1 \gg \omega_2$, then the rotation is unstable (see Figure \ref{4dpd3}; note that there is a complex intersection in the second diagram).  If $\omega_1 \ll \omega_2$, then the rotation is stable (see Figure \ref{4dpd4}).

	\end{enumerate}
	The conclusions of items 1, 2, 3 above coincide with the results of \cite{Marshall, Casu2, Ratiu}. Of course, the papers \cite{Marshall, Casu2, Ratiu} do not use the language of parabolic diagrams and give stability conditions in terms of some inequalities. However, the translation from the language of parabolic diagrams to the language of inequalities is straightforward.
	
	\par
			Note that there is a case with a tangency point in the upper half-plane (Figure \ref{4dpd5}) when Theorem \ref{stabThm} is not applicable. It is claimed in \cite{Marshall, Ratiu} that this rotation is unstable, however this conclusion seems to be incorrect. This follows from the results of \cite{AS} and can also be deduced from the bifurcation diagrams constructed by Oshemkov \cite{Oshemkov}.
		
	\end{example}
\begin{example}[Two-dimensional rotation] Suppose that there is only one plane of rotation $\Pi_1$, i.e. the body is rotating about a subspace of codimension two. 
Sort principal axes in order of increasing eigenvalues of $J$. Then Theorem \ref{stabThm} implies that the rotation is stable if and only if the plane of rotation is spanned by two adjacent axes: $\Pi_1 = \langle e_i, e_{i+1} \rangle$. This result can be viewed as a natural generalization of the Euler
theorem.
\end{example}
\begin{example}[Rotation with one fixed axis] Suppose that there is only fixed axis, i.e. $\Pi_0$ is one-dimensional. 
Sort principal axes in order of increasing eigenvalues of $J$. Assume that the fixed axis is in the even place (2nd, or 4th, or 6th, ...). Then the rotation is unstable. This result can also be viewed as a natural generalization of Euler's theorem
about instability of rotation about the middle axis.
\end{example}
\begin{example}[Spiegler's theorem]\label{spieler}
Below is the main result of the work \cite{Spiegler}, reformulated in terms of the present paper.
\begin{theorem}[Spiegler \cite{Spiegler}]\label{spThm}
 Consider a regular stationary rotation of a $2m$-dimensional rigid body. Sort principal axes in order of increasing eigenvalues of $J$. Assume that \begin{enumerate}\item all planes of rotation are spanned by two adjacent axes: $\Pi_1 = \langle e_1, e_2 \rangle,\Pi_2 = \langle e_3, e_4 \rangle, \dots $;\item $|\omega_1| > \dots > |\omega_n|$, or $|\omega_1| < \dots < |\omega_n|$.\end{enumerate} Then the rotation is stable. 
 \end{theorem}
 It is easy to see that Theorem \ref{stabThm} implies Theorem \ref{spThm}. Moreover, Theorem \ref{stabThm} implies that condition 2 of Theorem \ref{spThm} can be omitted, and there are much more stability cases not covered by the result of Spiegler (see e.g. Figure \ref{4dpd1}).\par
Note that Spiegler's approach to the problem is based on the method known as the Arnold enegy-Casimir method (see \cite{AEC, Arnold}). As he proves, condition of Theorem \ref{spThm} is necessary and sufficient for the Hessian of the energy to be positive-definite on the coadjoint orbit. By comparing Theorem \ref{spThm} with Theorem \ref{stabThm}, we conclude that for the majority of stable stationary rotations the Hessian of the energy is indefinite, so the energy-Casimir method fails. For these rotations, another Lyapunov function is needed to prove stability. Such a function can be explicitly found in small dimensions, as it was done in \cite{Marshall, Casu2, Casu}, however it is not clear how to proceed for general $n$. The method of the present paper allows to prove the existence of a Lyapunov function without finding it explicitly. 
\end{example}

\section{Bi-Hamiltonian structures and stability}\label{biHam}
In this section, basic definitions and theorems related to stability in bi-Hamiltonian systems are formulated. Most of them can be found in \cite{JGP, SBS}. Basic notions from Poisson geometry used throughout the section can be found in \cite{Zung}.
	\subsection{Basic notions}
	
\begin{definition}
	Two Poisson brackets on a manifold $\mathrm M$ are called \textit{compatible}, if any linear combination of them is a Poisson bracket again.
The \textit{Poisson pencil} generated by two compatible Poisson brackets $P_0, P_\infty$ is the set
\begin{align*}\Pi \defeq \{P_\lambda = P_0 - \lambda P_\infty\}_{\lambda \in \CP}.\end{align*}
A vector field $\mathbf X$ is bi-Hamiltonian with respect to a pencil $\Pi$ if it is hamiltonian with respect to all brackets of the pencil, i.e. for any $\lambda \in \CP$ there exists a (complex-valued) smooth function $H_\lambda$ such that
$$
\mathbf X = P_\lambda \diff H_\lambda.
$$
\end{definition}
The notion of a bi-Hamiltonian system was introduced by F.Magri \cite{Magri}, I.\,Gelfand and I.\,Dorfman \cite{GD}.
	\begin{remark}\label{complCotSpace}
		For complex values of $\lambda$, the bracket $P_\lambda$ should be treated as a complex-valued Poisson bracket on complex-valued functions. The corresponding Poisson tensor is a bilinear form on the complexified cotangent space at each point.
	\end{remark}
\begin{definition}
The \textit{rank of a pencil $\Pi$ at a point $x$} is the number
\begin{align*}
\rank \Pi(x) \defeq \max_{\lambda \in \CP} \rank P_\lambda(x).
\end{align*}
The \textit{rank of a pencil $\Pi$} (on a manifold $\mathrm M$) is the number
\begin{align*}
\rank \Pi \defeq \max_{x \in \mathrm M} \rank \Pi(x).
\end{align*}
\end{definition}
\begin{definition}
	The \textit{spectrum} of a pencil $\Pi$ at a point $x$ is the set
	\begin{align*}
	\Lambda_\Pi(x) \defeq \{ \lambda \in \CP \mid \rank P_\lambda(x) < \rank \Pi(x)\}.
	\end{align*}
\end{definition}
When $\Pi$ is fixed, the notation $\Lambda(x)$ is also used.\par \smallskip
By $\g_{\lambda}(x)$, denote the Lie algebra structure defined on $\Ker P_\lambda(x)$ by the linear part of $P_\lambda$ at the point $x$. The commutator in $\g_\lambda$ is given by
$$
[\xi, \eta]_\lambda \defeq \diff \{f,g\}_\lambda(x)
$$
where $\xi, \eta \in \Ker P_\lambda$, and $f,g$ are any smooth functions such that $\diff f(x) = \xi, \diff g(x) = \eta$.\par
The algebra $\g_\lambda(x)$ is mainly considered only for $\lambda \in \Lambda(x)$.\
\begin{remark}
	For $\lambda \in \RP$, the algebra $\g_\lambda$ is real. However, for complex values of $\lambda$, the space $\Ker P_\lambda(x)$ is a subspace of $\T^*_xM \otimes \Complex$,  and therefore $\g_\lambda$ is considered as a complex Lie algebra.
\end{remark}
Let $\mathbf X$ be a system which is bi-Hamiltonian with respect to $\Pi$, and let $x$ be an equilibrium point of $\mathbf X$. \begin{definition}\label{regEq}
Say that $x$ is \textit{regular} if the following condition holds:
	\begin{align*}
	\Ker P_\alpha(x) = \Ker P_\beta(x) \mbox{ for all } \alpha, \beta \notin \Lambda(x).
	\end{align*}
\end{definition}
\begin{remark}
Under some additional technical assumptions, regularity is equivalent to the following: $x$ is an equilibrium for \textbf{all} systems which are bi-Hamiltonian with respect to the pencil $\Pi$ (see \cite{biham, JGP}). Equilibria not satisfying this condition are normally unstable.

\end{remark}
For a regular equilibrium $x$, denote $\h(x) \defeq \Ker P_\alpha(x) = \Ker P_\beta(x)$.
 It is easy to see that $\h(x) \subset \g_\lambda(x)$ is a Lie subalgebra for each $\lambda$.

 \subsection{Spectral formula for bi-Hamiltonian systems}
 Let $x$ be a regular equilibrium of a bi-Hamiltonian system $\mathbf X = P_\lambda \diff H_\lambda(x)$. 
 Let $\xi \in \g_\lambda(x) = \Ker P_\lambda(x)$. Denote
$$
\ad_{\lambda}\xi(\eta) \defeq [\xi,\eta]_\lambda
$$
where $\LieBracket_\lambda$ is the commutator in $\g_\lambda(x)$. \par Note that $P_\lambda \diff H_\lambda(x) = \mathbf X(x) = 0$, so $\diff H_\lambda(x) \in \Ker P_\lambda(x) = \g_\lambda(x)$. Consequently, the operator $\ad_{\lambda}\diff H_\lambda(x)$ is well-defined.
\begin{statement}\label{invOfK}
	The subalgebra $\h(x) \subset \g_\lambda(x)$ is invariant with respect to $\ad_{\lambda}\diff H_\lambda(x)$, i.e.
	$$
	[\diff H_\lambda(x), \h(x)]_\lambda \subset \h(x).
	$$
\end{statement}
\begin{proof}
	Let $\xi \in \h(x)$. By definition,
	$$[\diff H_\lambda(x), \xi]_\lambda = \diff \{H_\lambda, g\}_\lambda(x)$$ where $g$ is any function such that $\diff g(x) = \xi$. Further,
	$$ \{H_\lambda, g\}_\lambda =  \mathbf X(g) = \{H_\alpha, g\}_\alpha$$
	for any $\alpha \in \CP$. So,
	$$[\diff H_\lambda(x), \xi]_\lambda = \diff  \{H_\alpha, g\}_\alpha(x) = [\diff H_\alpha, \xi]_\alpha.$$
	If $\alpha \notin \Lambda(x)$, then  $[\diff H_\alpha, \xi]_\alpha \in \Ker P_\alpha(x) = \h(x)$, so $[\diff H_\lambda(x), \xi]_\lambda \in \h(x)$, q.e.d.
\end{proof}

Let $x$ be a regular equilibrium point of a bi-Hamiltonian system. Then all symplectic leaves of generic brackets $P_\alpha, \alpha \notin \Lambda(x)$ are tangent to each other. Denote their common tangent space by $\mathrm T(x)$.\par
The following statement is used to find the spectrum of a bi-Hamiltonian system linearized at a regular equilibrium point.
\begin{lemma}
\label{specLemma}
Suppose that $\Pi$ is a Poisson pencil on a finite-dimensional manifold, and $\mathbf X = P_\lambda \diff H_\lambda$ is a system which is bi-Hamiltonian with respect to $\Pi$. Let $x$ be a regular equilibrium of $\mathbf X$. 

 Then the spectrum of the linearization of $\mathbf X$ at $x$ restricted to $\mathrm T(x)$ is given by
 $$
 \sigma(\diff \mathbf X\mid_{\mathrm T(x)}) = \bigcup_{\lambda \in \Lambda(x)} \sigma\left((\ad_\lambda\,\diff H_\lambda(x))\mid_{\g_\lambda(x) / \h(x)}\right)
 $$
  where $\sigma(P)$ stands for the spectrum of the operator $P$.
 \end{lemma}
Note that the restriction of $\ad_\lambda\,\diff H_\lambda(x)$ to ${\g_\lambda(x) / \h(x)}$ is well-defined since $\h(x)$ is invariant with respect to $\ad_{\lambda}\diff H_\lambda(x)$ (see Proposition \ref{invOfK}).\par
\smallskip 
The proof of Lemma \ref{specLemma} easily follows from the results of \cite{JGP, SBS}.
\subsection{Linearization of a Poisson pencil and nonlinear stability}
\begin{definition}
	Let $\g$ be a (real or complex) Lie algebra, and let $\mathfrak B$ be a skew-symmetric bilinear form on $\g$. 
	Then $\mathfrak B$ can be considered as a Poisson tensor on the dual space $\g^{*}$. Assume that the corresponding bracket is compatible with the Lie-Poisson bracket. In this case the Poisson pencil $\Pi(\g, \mathfrak B)$ generated by these two brackets is called the \textit{linear pencil} associated with the pair $(\g, \mathfrak B)$.
\end{definition}

\begin{statement}
	A form $\mathfrak B$ on $\g$ is compatible with the Lie-Poisson bracket if and only if this form is a Lie algebra $2$-cocycle, i.e. 
	\begin{align*}
		\diff \mathfrak B(\xi,\eta,\zeta) \defeq \mathfrak B([\xi,\eta],\zeta) + \mathfrak B([\eta,\zeta],\xi) + \mathfrak B([\zeta,\xi],\eta) = 0
	\end{align*}
	for any $\xi, \eta, \zeta \in \g$.
\end{statement}
Below is the central construction of the theory discussed in the present section. Let $\Pi$ be an arbitrary Poisson pencil on a manifold $\mathrm M$, and $x \in \mathrm M$. As before, denote the Lie algebra on $\Ker P_\lambda(x)$ by $\g_{\lambda}(x)$.
It turns out that apart from the Lie algebra structure, $\g_{\lambda}$ carries one more additional structure.
\begin{statement}\label{PAlphaOnKerPLambda}\quad\par
\begin{enumerate}
	\item For any $\alpha$ and $\beta$  the restrictions of $P_{\alpha}(x), P_{\beta}(x)$ on $\g_{\lambda}(x)$ coincide up to a constant factor.
	\item The $2$-form $P_{\alpha}|_{\g_{\lambda}}$ is a $2$-cocycle on $\g_{\lambda}$.
\end{enumerate}
\end{statement}

Consequently, $P_{\alpha}|_{\g_{\lambda}}$ defines a linear Poisson pencil on $\g_{\lambda}^{*}$. Since $P_{\alpha}|_{\g_{\lambda}}$ is defined up to a constant factor, the pencil is well-defined. Denote this pencil by $\diff_{\lambda} \Pi(x)$.
\begin{definition}
	The pencil $\diff_{\lambda} \Pi(x)$ is called the \textit{$\lambda$-linearization} of the pencil $\Pi$ at $x$.
\end{definition}

Now, let $\mathfrak B$ be a $2$-cocycle on a Lie algebra $\g$.
For an arbitrary element $\nu \in \Ker \mathfrak B$, define the bilinear form $$\mathfrak B_\nu(\xi, \eta) \defeq \mathfrak B([\nu, \xi], \eta).$$ The cocycle identity implies that this form is symmetric. Furthermore, $\Ker \mathfrak B_\nu \supset \Ker \mathfrak B$, therefore $\mathfrak B_\nu$ is a well-defined symmetric form on the vector space $\g / \Ker \mathfrak B$.

 \begin{definition}\label{compDef}
 A linear pencil $\Pi(\g, \mathfrak B)$ is \textit{compact} if there exists $\nu \in \zenter(\Ker \mathfrak B)$ such that $\mathfrak B_{\nu}$ is positive-definite on $\g / \Ker \mathfrak B$. 
\end{definition}
\begin{remark}The notation $\zenter(\g)$ stands for the center of the Lie algebra $\g$. \end{remark} 
Definition \ref{compDef} is motivated by the following statement.
\begin{statement}\label{compcomp}
	Any linear pencil on a compact semisimple Lie algebra is compact.
\end{statement}
\begin{proof}	 Let $\g$ be a compact semisimple Lie algebra. Since $\Hom^2(\g) = 0$, any cocycle $\mathfrak B$ on $\g$ has the form $\mathfrak B(\xi, \eta) = \langle \zeta,[\xi, \eta] \rangle$, where $\langle\, , \rangle$ is the Killing form, $\zeta \in \g$. Take $\nu = \zeta$. Note that $\Ker \mathfrak B$ is the centralizer of $\nu$, so $\nu \in \zenter(\Ker \mathfrak B)$. Further,
$$
\mathfrak B_\nu(\xi, \xi) =\langle \nu, [[\nu, \xi], \xi] \rangle = -\langle [\nu, \xi], [\nu, \xi] \rangle > 0,
$$
so the pencil is compact.
\end{proof}
Another motivation for Definition \ref{compDef} is the following fact: let a system $\mathbf X$ be bi-Hamiltonian with respect to a compact linear pencil; then all trajectories of $\mathbf X$ are bounded.\par
In the present paper, there will be non-trivial examples of compact linear pencils on non-compact Lie algebras $\mathfrak{u}(p,q)$ and $\mathfrak{u}(p,q) \ltimes \Complex^{p+q}$ arising as $\lambda$-linearizations of the pencil related to the multidimensional rigid body (see Appendix).

\begin{definition}\label{defDiag}
	A pencil $\Pi$ is called diagonalizable at $x$ if
	\begin{align}\label{diagCond}
	\dim \Ker \left( P_\alpha(x)|_{P_\lambda(x)}\right) = \corank \Pi(x) \mbox{ for all } \lambda \in \Lambda_\Pi(x), \alpha \neq \lambda.
	\end{align}
\end{definition}
Note that if (\ref{diagCond}) is satisfied for some $\alpha \neq \lambda$, then it is satisfied for any $\alpha \neq \lambda$ (see Proposition \ref{PAlphaOnKerPLambda}).\par\smallskip
The following theorem is used to prove nonlinear stability for a bi-Hamiltonian system.
\begin{theorem}
\label{bhStabThm}
Suppose that $\Pi$ is a Poisson pencil on a finite-dimensional manifold, and $\mathbf X$ is bi-Hamiltonian with respect to $\Pi$. Let $x$ be an equilibrium point of $\mathbf X$. Assume that  
 \begin{enumerate}
	\item $\rank \Pi(x) = \rank \Pi$.
	\item The equilibrium $x$ is regular.
		\item The spectrum of $\Pi$ at $x$ is real: $\Lambda_\Pi(x) \subset \RP$.

	\item The pencil $\Pi$ is diagonalizable at $x$.
	\item For each $\lambda \in \Lambda_\Pi(x)$ the $\lambda$-linearization $\diff_\lambda \Pi(x)$ is compact.
 \end{enumerate}
 Then $x$ is Lyapunov stable.
 \end{theorem}
 See \cite{JGP} for the proof.
\begin{remark}
	The idea of the proof can be explained as follows. A bi-Hamiltonian system automatically possesses a large number of first integrals: these are the Casimir functions of all brackets of the pencil. The Hessians of these functions are controlled by linear parts of the corresponding brackets. The conditions of the theorem allow to show that there exists a linear combination of Casimir functions with a positive-definite Hessian, so that this combination can be used as a Lyapunov function.
	
\end{remark}

\subsection{Generalized stability theorem}
In this section, a stronger stability result is formulated which allows to proceed for those points where $\rank \Pi(x) < \rank \Pi$. The condition $\rank \Pi(x) = \rank \Pi$ can only be omitted under some additional technical assumptions.\par
Let $\mathbf X$ be a system which is bi-Hamiltonian with respect to a pencil $\Pi$, and let $x$ be a regular equilibrium point of $\mathbf X$. Then for each $\alpha \notin \Lambda(x)$, the linear part of $P_\alpha(x)$ defines a natural Lie algebra structure on $\h(x) = \Ker P_\alpha(x)$. For $\lambda \in \Lambda(x)$, there is a strict inclusion $\h(x) \subset \Ker P_\lambda(x)$. However, since $P_\lambda(x)$ is a linear combination of $P_\alpha(x)$ and $P_\beta(x)$  for any $\alpha \neq \beta \in \CP$, the subspace $\h(x)$ is a subalgebra in $ \Ker P_\lambda(x)$ for any $\lambda$. Thus, $\h(x)$ carries a structure of a Lie pencil. Denote by $\zenter_\alpha (\h(x))$ the center of $\h(x)$ with respect to the Lie structure $\LieBracket_\alpha$.
\begin{definition}\label{regEq2}
Say that $x$ is \textit{strongly regular} if it is regular, and $$\zenter(\h(x)) \defeq \zenter_\alpha (\h(x))$$ does not depend on $\alpha$.
\end{definition}
\begin{remark}
	The center of the kernel is important by the following reason: if $f$ is a Casimir function of $P$, then $\diff f(x) \in \zenter(\Ker P(x))$. Moreover, if the transverse Poisson structure to $P$ at the point $x$ is linearizable, then the differentials of Casimir functions span $\zenter(\Ker P(x))$.
\end{remark}
\begin{definition}\label{finePencil}
A pencil $\Pi$ is called \textit{fine} at a point $x$ if there exists $\alpha \notin \Lambda(x)$ and an open neighborhood $U \ni \alpha$ such that
	\begin{enumerate}
		\item  for each $\beta \in U$ the transverse Poisson structure to $P_\beta$ at the point $x$ is linearizable, and its linear part is compact;
		\item for any $f_\alpha \in \zenter(P_\alpha)$ there exists a family $f_\beta$ depending continuously on $\beta$ and defined for $\beta \in U$ such that $f_\beta \in \zenter(P_\beta)$, i.e. any Casimir function of $P_\alpha$ can be ``approximated'' by Casimir functions of nearby brackets of the pencil.
		\end{enumerate}
\end{definition}
\begin{remark}The notation $\zenter(P)$ stands for the set of (local) Casimir functions of the Poisson bracket $P$. \end{remark} 
The following theorem is a generalization of Theorem \ref{bhStabThm}.
\begin{theorem}
\label{genStabThm}
Suppose that $\Pi$ is a Poisson pencil on a finite-dimensional manifold, and $\mathbf X$ is a dynamical system which is bi-Hamiltonian with respect to $\Pi$. Let $x$ be an equilibrium point of $\mathbf X$. Assume that  
 \begin{enumerate}

	\item The pencil $\Pi$ is fine at $x$.

	\item The equilibrium point $x$ is strongly regular.
	\item The spectrum of $\Pi$ at $x$ is real: $\Lambda_\Pi(x) \subset \RP$.
	
	\item The pencil $\Pi$ is diagonalizable at $x$.
	\item For each $\lambda \in \Lambda_\Pi(x)$ the $\lambda$-linearization $\diff_\lambda \Pi(x)$ is compact.
 \end{enumerate}
 Then $x$ is Lyapunov stable.
 \end{theorem}
 It is easy to see that if $\rank \Pi(x) = \rank \Pi$, then the pencil $\Pi$ is fine at $x$, and regularity is equivalent to strong regularity. So, Theorem \ref{genStabThm} is a generalization of Theorem \ref{bhStabThm}.
 The proof of Theorem \ref{genStabThm} repeats the proof of Theorem \ref{bhStabThm}.

 \section{Bi-Hamiltonian structure of the multidimensional rigid body}\label{biHamMRB}
  Denote the standard Lie bracket on $\so(n)$ by $\LieBracket_\infty$ and the corresponding Lie-Poisson bracket on $\so(n)^{*}$ by $\{\,,\}_{\infty}$. The latter is given by
 $$
 \{f,g\}_{\infty}(M) \defeq \langle M, [\diff f, \diff g]_\infty\rangle 
 $$
for $M \in \so(n)^{*} \simeq \so(n)$ and $f,g \in \Cont^\infty(\so(n)^*)$. By $\langle \,,\rangle$ we denote the Killing form $$\langle X,Y\rangle = \Tr XY,$$ and $\so(n)^*$ is identified with $\so(n)$ by means of this form.\par
 The following was observed by Arnold \cite{Arnold}.
 \begin{statement}
 	The equations (\ref{eae}) are Hamiltonian with respect to the bracket $\PoissonBracket_\infty$.  The Hamiltonian is given by the kinetic energy
	$$
		H_\infty \defeq \frac{1}{2}\langle \Omega, M \rangle.
	$$
 \end{statement}
Now introduce a second operation on $\so(n)$ defined by
$$
	[X,Y]_{0} \defeq XJ^{2}Y - YJ^{2}X. 
$$
  \begin{statement}\quad\par
  \begin{enumerate}
  \item
 	$\LieBracket_{0}$ is a Lie bracket compatible with the standard Lie bracket. In other words, any linear combination of these brackets defines a Lie algebra structure on $\so(n)$.
	\item The corresponding Lie-Poisson bracket $\{\,,\}_{0}$ on $\so(n)^{*}$ given by
  $$
 \{f,g\}_{0} \defeq \langle M, [\diff f,\diff g]_0\rangle.
 $$ is  compatible with the Lie-Poisson bracket $\{\,,\}_{\infty}$.
	\end{enumerate}
 \end{statement}
Consequently, a Lie pencil is defined on $\so(n)$, and a Poisson pencil is defined on $\so(n)^{*}$.  Write down these pencils in the form
\begin{align}\label{LiePencil}
		[X,Y&]_\lambda = [X,Y]_0 - \lambda[X,Y]_\infty =X (J^{2} - \lambda E) Y - Y (J^{2} - \lambda E) X,
		\end{align}
		for $X,Y \in \so(n)$,
		and
		\begin{align}\label{PoissonPencil}
 		\{f,g\}_{\lambda}&= \{f,g\}_{0} - \lambda\{f,g\}_{\infty} =  \langle M, \diff f (J^{2} - \lambda E) \diff g - \diff g (J^{2} - \lambda E) \diff f\rangle
\end{align}
for $M \in \so(n)^{*}$ and $f,g \in \Cont^\infty(\so(n)^*)$.\par \smallskip
The Poisson tensor corresponding to the bracket $\{\,,\}_{\lambda}$ reads
		\begin{align}\label{PoissonTensor}
 		P_{\lambda}(M)(X,Y) &=  \langle M, X (J^{2} - \lambda E)Y - Y(J^{2} - \lambda E) X\rangle
\end{align}
where $M \in \so(n)^*$, and $X, Y \in \mathrm{T}^*_M(\so(n)^*) = \so(n)$.
\begin{statement}[Bolsinov \cite{Bolsinov2, biham1}]
The system  (\ref{eae}) is Hamiltonian with respect to any bracket $\{\,,\}_{\lambda}$, so it is bi-Hamiltonian. The Hamiltonian is given by
\begin{align}\label{hLambdaFormula}
H_{\lambda} \defeq -\frac{1}{2}\langle (J + \sqrt{\lambda} E)^{-1}\Omega(J+\sqrt{\lambda} E)^{-1},M \rangle.
\end{align}
\end{statement}
\begin{remark}
	The matrix $J + \sqrt{\lambda} E$ is invertible for any $\lambda$ if the proper value of the square root is chosen.

	Note that the function $H_\lambda$ written here is different from the one given by Bolsinov. The difference is a Casimir function of $P_\lambda$.
\end{remark}

\section{Proof of Theorem 1: instability}\label{instabSect}
The proof consists of the following steps:
\begin{enumerate}
 \item check that a regular stationary rotation (in the sense of Definition \ref{regDef}) is a regular equilibrium (in the sense of Definition \ref{regEq}), so that Lemma \ref{specLemma} can be applied (Section \ref{sectReg});
 \item describe the spectrum $\Lambda(M)$ (Section \ref{sectSpec});
 \item describe the adjoint operators $\ad_\lambda$ (Section \ref{sectLin});
 \item find the spectrum of the linearized system using Lemma  \ref{specLemma} (Section \ref{sectLinSpec}).
\end{enumerate}
Fix some notation which is used throughout the proofs.\par It is only regular stationary rotations which are considered. So, assume that there exists an orthonormal basis such that $J$ is diagonal, while $\Omega$ and $M$ are block-diagonal with two-by-two blocks on the diagonal (Definition \ref{regDef}).
Denote by $\lambda_{i}$ the diagonal elements of $J$ \textbf{in this basis}. Note that this means that $\lambda_{i}$ are possibly different for different rotations. However, they are unique up to a permutation and coincide with the eigenvalues of $J$.\par
By $\omega_{i}$'s, denote the non-zero entries of the matrix $\Omega$ as in (\ref{omegaForm}).
By $m_{i} = (\lambda_{2i-1} + \lambda_{2i})\omega_{i}$, denote the non-zero entries of the matrix $M$. The notation $M_i$ stands for the diagonal two-by-two blocks of $M$, i.e.
\begin{align}\label{MBlock}
	M_i \defeq \left(\begin{array}{cc}0 & m_i \\-m_i & 0\end{array}\right).
\end{align}
The number $n$ stands for the dimension of the body, and $m$ stands for the number of non-zero $\omega_i$'s, that is for the number of two-dimensional planes in the decomposition (\ref{decomp}).\par
For a fixed $\lambda$, let $A \defeq J^2 - \lambda E$ if $\lambda \neq \infty$ or $A \defeq E$ otherwise. By $a_i$, denote the diagonal entries of the matrix $A$. Clearly, $a_i = \lambda_i^2 - \lambda$ if $\lambda \neq \infty$, and $a_i = 1$ otherwise. It is also convenient to represent $A$ as
\begin{align}\label{AForm}
A = \left(\begin{array}{cccccc}A_1 &  &  &  &  &  \\ & \ddots &  &  &  &  \\ &  & A_m &  &  &  \\ &  &  & a_{2m+1} &  &  \\ &  &  &  & \ddots &  \\ &  &  &  &  & a_n\end{array}\right),
\end{align}
where $A_{i}$ are two-by-two diagonal matrices, and $a_i$ are numbers. \par 
Further, extend the definition of $\chi_i(x)$ given by (\ref{chiFormula}) to the point $\infty$.
$$
	\chi_{i}(\infty) \defeq  \frac{1}{\omega_{i}^{2}( \lambda_{2i-1} +  \lambda_{2i} )^{2}} = \frac{1}{m_i^2}.
$$
Note that for each $\lambda \in \CP$ the following equality holds.
\begin{align}\label{chiFla2}
\chi_i(\lambda) = \frac{a_{2i-1}a_{2i}}{m_i^2}.
\end{align}
 \subsection{Regularity}\label{sectReg}
Let $M$ be a regular stationary rotation (in the sense of Definition \ref{regDef}). Find a basis such that $J$ is diagonal, and $M$ is block-diagonal. Introduce the following subspaces:
 
 \begin{itemize}
 	\item $K \subset \so(n)$ is generated by $\{E_{2i-1, 2i} - E_{2i, 2i-1}\}_{i = 1, \dots, m}$ and $\{E_{ij} - E_{ji}\}_{2m < i < j \leq n}$.
 	\item $V_{ij} \subset \so(n)$ is generated by 
	$E_{2i-1,2j-1} - E_{2j-1,2i-1}$, $E_{2i-1,2j} - E_{2j,2i-1}$, $E_{2i,2j-1} - E_{2j-1,2i}$, $E_{2i,2j} - E_{2j,2i}$.
	\item $W_{ij} \subset \so(n)$ is generated by $E_{2i-1, j} - E_{j, 2i-1}$, $E_{2i,j} - E_{j,2i}$.
  \end{itemize}
Clearly, the following vector space decomposition holds
\begin{align}\label{sonDec}
 	\so(n) = K \oplus \bigoplus\limits_{1 \leq i < j \leq m}V_{ij} \oplus 
	\bigoplus_{\substack{1 \leq i \leq m, \\ 2m <j \leq n}} W_{ij}.
\end{align}
\begin{statement}\label{KcommonKer}
	The space $K$ belongs to the common kernel of all brackets of the pencil at the point $M$. All spaces $V_{ij}, W_{i}$ are mutually orthogonal with respect to all brackets of the pencil.
\end{statement}
\begin{proof}
	Use (\ref{PoissonTensor}).
\end{proof}
Proposition \ref{KcommonKer} implies that the rank of a bracket $P_\lambda$ drops if and only if this bracket is degenerate on one of the spaces $V_{ij}$ or $W_{ij}$. Calculate the forms $P_\lambda$ on these spaces.\par
Identify $V_{ij}$ with the space of two-by-two matrices, and $W_{ij}$ with $\R^{2}$. 
Let the matrices $M_i$ be defined by (\ref{MBlock}). Let also the numbers $a_j$ and the matrices $A_i$ be defined by (\ref{AForm}).

\begin{statement}\label{PLambdaRestr}
	The form $P_{\lambda}$ restricted to $V_{ij}$ reads
	$$
		P_{\lambda}(X,Y) = 2\Tr( M_{i}XA_{j}Y^{\mathrm{t}} + M_{j}X^{\mathrm{t}}A_{i}Y).
	$$
	
	The form  $P_{\lambda}$ restricted to $W_{ij}$ reads
	$$
		P_{\lambda}(v,w) = -2a_{j} v^{\mathrm{t}}M_{i}w. 
	$$	
\end{statement}
\begin{proof}
	Use (\ref{PoissonTensor}).
\end{proof}
Now calculate $P_{\lambda}$ on $V_{ij}$ in coordinates. Let
$$
   X = \left(\begin{array}{cc}a & b \\c & d\end{array}\right) \in V_{ij}, \quad Y = \left(\begin{array}{cc}e & f \\g & h\end{array}\right) \in V_{ij}.
$$

Explicit calculation shows that
	\begin{align*}
		P_{\lambda}(X,Y) &= 2(m_{i}a_{2j-1}c +m_{j}a_{2i-1}b)e + 2(m_{j}a_{2i}d - m_{i}a_{2j-1}a)g + \\
		 &+ 2(m_{i}a_{2j}d -m_{j}a_{2i-1}a)f  - 2(m_{i}a_{2j}b + m_{j}a_{2i}c)h.
	\end{align*}
Consequently, $X \in \Ker P_{\lambda}$ if and only if
	\begin{align*}
\begin{cases}
	m_{i}a_{2j-1}c +m_{j}a_{2i-1}b = 0,\\
	m_{j}a_{2i}d - m_{i}a_{2j-1}a = 0,\\
	m_{i}a_{2j}d -m_{j}a_{2i-1}a = 0,\\
	m_{i}a_{2j}b + m_{j}a_{2i}c = 0.
\end{cases}
\end{align*}
This system can be split into two two-by-two systems, and the determinant of both of them equals
$$
\mathrm det = m_{j}^{2}a_{2i-1}a_{2i} - m_{i}^{2}a_{2j-1}a_{2j}. 
$$ 
So, the following is true.
\begin{statement}\label{degonVij}
 $P_{\lambda}$ is degenerate on $V_{ij}$ if and only if \begin{align}\label{specEq}m_{j}^{2}a_{2i-1}a_{2i} - m_{i}^{2}a_{2j-1}a_{2j} = 0.\end{align}

If  $P_{\lambda}$ is degenerate on $V_{ij}$, then its kernel is given by
\begin{align}\label{alphabeta}
X =  \left(\begin{array}{cc}
	\alpha m_{j}a_{2i} & \beta m_{i}a_{2j-1},\\
	-\beta m_{j}a_{2i-1} &\alpha m_{i}a_{2j-1}
	\end{array}\right)
\end{align}
where $\alpha$ and $\beta$ are arbitrary numbers.

\end{statement}
Now study $P_{\lambda}$ on $W_{ij}$. 
\begin{statement}\label{degonWij}
 $P_{\lambda}$ is degenerate (and, consequently, zero) on $W_{ij}$ if and only if $\lambda = \lambda_{j}^{2}$.
\end{statement}
\begin{proof}
Use Proposition \ref{PLambdaRestr}.
\end{proof}
\begin{statement}\label{reg1}
	The intersection of kernels of all brackets of the pencil is exactly $K$. For almost all brackets the kernel is exactly $K$. Thus, $M$ is a regular equilibrium (in the sense of Definition \ref{regEq}), and $\h(M) = K$.
\end{statement}
\begin{proof}
	Only finite number of brackets are degenerate on each of the spaces $V_{ij}$ and $W_{ij}$ (see Propositions \ref{degonVij}, \ref{degonWij}).
\end{proof}

\subsection{Description of the spectrum $\Lambda(\M)$}\label{sectSpec}
\begin{statement}\label{specProp}
	Let $M$ be a regular stationary rotation. Then $\Lambda(\M)$ is the set of horizontal coordinates of intersection points in the parabolic diagram of $M$.
\end{statement}
\begin{proof}
	By Proposition \ref{degonVij}, the bracket $P_\lambda$ is degenerate on $V_{ij}$ if and only if $$m_{j}^{2}a_{2i-1}a_{2i} - m_{i}^{2}a_{2j-1}a_{2j} = 0.$$ This equality can be rewritten as (see (\ref{chiFla2}))
	$$
		\chi_{i}(\lambda) = \chi_{j}(\lambda),
	$$	
	which means that $\lambda$ is the horizontal coordinate of the intersection point of two parabolas $y = \chi_i(x)$ and $y = \chi_j(x)$.
	Further, by Proposition \ref{degonWij}, $P_{\lambda}$ is degenerate on $W_{ij}$ if and only if $\lambda = \lambda_{j}^{2}$, which means that $\lambda$ is the horizontal coordinate of the intersection point of the vertical line $x = \lambda_{j}^2$ with any parabola.
\end{proof}

\subsection{Description of adjoint operators}\label{sectLin}	
Compute the restriction of the operator $\ad_\lambda\,\diff H_\lambda$ to the space  ${\Ker P_\lambda / K}$.
	Using Proposition \ref{KcommonKer}, the kernel $\Ker P_{\lambda}(M)$ can be decomposed in the following way
	$$
		\Ker P_{\lambda} = K \oplus  \bigoplus\limits_{1 \leq i < j \leq m} V_{ij}(\lambda) \oplus 
	\bigoplus_{\substack{1 \leq i \leq m, \\ 2m <j \leq n}}  W_{ij}(\lambda),
	$$
where \begin{align}\label{vijLambda} V_{ij}(\lambda) \defeq  \Ker\left( P_{\lambda}\mid_{  V_{ij}}\right) \subset V_{ij}, \quad W_{ij}(\lambda)  \defeq  \Ker\left( P_{\lambda}\mid_{  W_{ij}}\right) \subset W_{ij}.\end{align} 

The space $\Ker P_\lambda / K$ is decomposed as
	$$
		\Ker P_{\lambda} / K =  \bigoplus\limits_{1 \leq i < j \leq m} V_{ij}(\lambda) \oplus 
	\bigoplus_{\substack{1 \leq i \leq m, \\ 2m <j \leq n}}  W_{ij}(\lambda),
	$$

To compute $\ad_\lambda\,\diff H_\lambda$, note that $P_\lambda$ is a linear bracket, so the commutator in $\g_\lambda$ is simply the restriction of the bracket $\LieBracket_\lambda$  given by (\ref{LiePencil}) to the space $ \Ker P_\lambda$, so it is given by
	\begin{align}\label{BracketOnTheKernel}
		[X,Y]_{\lambda} = XAY - YAX.
	\end{align}
	
	Formula (\ref{hLambdaFormula}) implies that
	\begin{align}\label{dhLambdaFormula}
		\diff H_\lambda =  -(J + \sqrt{\lambda} E)^{-1}\Omega(J+\sqrt{\lambda} E)^{-1},
	\end{align}
	so 
	$\diff H_\lambda$ is a block-diagonal matrix with two-by-two blocks on the diagonal. Denote the space of such matrices by $L$. 
	Clearly,
	$
	L \subset K \subset \Ker P_\lambda
	$.
	 Since $\diff H_\lambda \in L$, it suffices to describe the operators $\ad_\lambda X$ for $X \in L$.
\begin{statement}
	 For any $X \in L$, the spaces $ V_{ij}(\lambda)$ and $W_{ij}(\lambda)$ are invariant with respect to the operator $\ad_{\lambda}X$.
\end{statement}
\begin{proof}
	Let $X \in L$ and $ Y \in V_{ij}$. Then, using  (\ref{BracketOnTheKernel}), show that $[X,Y]_\lambda \in V_{ij}$, which means that $\ad_\lambda X(V_{ij}) \subset V_{ij}$.
	 Further, $ V_{ij}(\lambda) = V_{ij} \cap \Ker P_{\lambda}$, and $\Ker P_{\lambda}$ is invariant with respect to $\ad_\lambda X$, so $ V_{ij}(\lambda)$ is invariant as the intersection of two invariant subspaces. The proof for $ W_{ij}(\lambda)$ is the same.
\end{proof}
Represent an element $X \in L$ as 
	$$
	X = \left(\begin{array}{cccc}0 & x_1 &   \\-x_1 & 0 &   \\ &  &  \ddots\end{array}\right).
	$$
\begin{statement}\label{eigen1}
	Let $ V_{ij}(\lambda) \neq 0$. Then 
	 the operator $\ad_\lambda{X}$, being written in coordinates $\alpha, \beta$ given by (\ref{alphabeta}), reads
	\begin{align*}
\left(\begin{array}{cc}0 & -a_{2i-1}(x_{i} - \dfrac{m_{j}}{m_{i}}x_{j}) \\a_{2i}(x_{i} - \dfrac{m_{j}}{m_{i}}x_{j}) & 0\end{array}\right).
\end{align*}

\end{statement}
\begin{proof}
Use (\ref{BracketOnTheKernel}).

\end{proof}

\begin{statement}\label{eigen2}
	Let $ W_{ij}(\lambda) \neq 0$. Then, after the natural identification of  $ W_{ij}(\lambda)=W_{ij}$ with $\R^2$, the matrix of $\ad_\lambda{X}$ reads
$$
 \left(\begin{array}{cc}0 & x_ia_{2i} \\-x_{i}a_{2i-1} & 0\end{array}\right).
$$

\end{statement}
\begin{proof}
Use (\ref{BracketOnTheKernel}).
\end{proof}

\begin{statement}\label{eigene}\quad\par
	\begin{enumerate} \item Let $ V_{ij}(\lambda) \neq 0$. Then 
	the eigenvalues of $\ad_\lambda{X}$ restricted to $ V_{ij}(\lambda)$ are
	$\pm \nu_{ij}(X)$, where
\begin{align*}
\nu_{ij}(X) &= \sqrt{-\chi_{i}(\lambda)}(m_{j}x_{j} -m_{i}x_{i}).
\end{align*}
\item 	Let $ W_{ij}(\lambda) \neq 0$. Then the eigenvalues of $\ad_\lambda{X}$ restricted to $ W_{ij}(\lambda)$ are
	$\pm \mu_{i}(X)$, where
	$$
	\mu_{i}(X) = \sqrt{-\chi_{i}(\lambda)}m_{i}x_{i}.
	$$

\end{enumerate}
\end{statement}
\begin{proof}
Use Propositions \ref{eigen1}, \ref{eigen2}.
\end{proof}

\subsection{Spectrum of the linearized system}\label{sectLinSpec}

Let $M$ be a regular stationary rotation. Then the symplectic leafs of generic brackets $P_\alpha, \alpha \notin \Lambda(M)$ are tangent at $M$. By $\mathrm T(M)$ denote their common tangent space. By $\diff \mathbf X$ denote the linearization of (\ref{eae}) at $M$.
\begin{statement}\label{spectra}
	Let $M$ be a regular stationary rotation. Let also $\lambda_{ij}^{(1)}, \lambda_{ij}^{(2)}$ be two roots of the equation $\chi_{i}(\lambda) = \chi_{j}(\lambda)$.
	Then the spectrum of $\diff \mathbf X\mid_{\mathrm T(M)}$ is $$\sigma(\diff  \mathbf X\mid_{\mathrm T(M)}) = \{ \pm \sigma_{ij}^{(k)} \}_{1 \leq i <j \leq m}^ {k = 1,2} \cup \{\pm \tau_{ij} \}_{1 \leq i \leq m}^{ 2m+1 \leq j \leq n}$$ where
	\begin{align*}
		\sigma_{ij}^{(k)} = \frac{1}{\sqrt{-\chi_{i}\left(\lambda_{ij}^{(k)}\right)}}&\left( \frac{\lambda_{ij}^{(k)} + \lambda_{2i-1}\lambda_{2i}}{ \lambda_{2i-1} + \lambda_{2i}} - \frac{\lambda_{ij}^{(k)} + \lambda_{2j-1}\lambda_{2j}}{ \lambda_{2j-1} + \lambda_{2j}}   \right) \mbox{ \quad if } \lambda_{ij}^{(k)} \neq \infty,\\
		 \sigma&_{ij}^{(k)} =  (\omega_j - \omega_i) \sqrt{-1} \mbox{\quad if } \lambda_{ij}^{(k)} = \infty,  \\
		 \end{align*}
		 and
		 \begin{align*}
		 		 \tau_{ij} =& \frac{1}{\sqrt{-\chi_{i}\left(\lambda_{j}^{2}\right)}}  \frac{(\lambda_{j} - \lambda_{2i-1})(\lambda_{j} - \lambda_{2i})}{ \lambda_{2i-1} + \lambda_{2i}}.
		\end{align*}

\end{statement}
\begin{proof}
	Use Lemma \ref{specLemma}, Proposition \ref{eigene} and formula (\ref{dhLambdaFormula}).
\end{proof}
\begin{remark}
	It is also possible to find the spectrum of the linearized system explicitly, without introducing the bi-Hamiltonian structure. However, the bi-Hamiltonian framework is essential for the proof of the stability part of Theorem \ref{stabThm} (Section \ref{stabSect}), so it seems to be better to prove both parts using the same philosophy. At the same time, the bi-Hamiltonian approach is simpler from the computational viewpoint.
	Also note that Lemma \ref{specLemma} allows to find the spectrum for \textit{all} systems bi-Hamiltonian with respect to $\Pi$ \textit{at once}.

\end{remark}
\subsection{Completion of the proof}
For simplicity assume that all eigenvalues of $\Omega$ are distinct. Suppose that there is at least one intersection in the parabolic diagram which is either complex or belong to the lower half-plane. Then Proposition \ref{spectra} shows that $\diff \mathbf X$ has an eigenvalue with a non-zero real part unless
\begin{align}\label{exoticCondition}
 \frac{\lambda_{ij}^{(k)} + \lambda_{2i-1}\lambda_{2i}}{ \lambda_{2i-1} + \lambda_{2i}} - \frac{\lambda_{ij}^{(k)} + \lambda_{2j-1}\lambda_{2j}}{ \lambda_{2j-1} + \lambda_{2j}} = 0.
\end{align}
A simple computation shows that (\ref{exoticCondition}) implies the equality $\omega_i^2 =  \omega_j^2$. If all eigenvalues of $\Omega$ are distinct, then this is not possible, so $\diff \mathbf X$ has an eigenvalue with a non-zero real part, and the equilibrium is unstable.\par 
The sketch of the proof in the case when $\Omega$ has multiple eigenvalues is as follows. Assume that \eqref{exoticCondition} is satisfied. Then $\lambda_{ij}^{(k)}$ is a real number. So, we only need to consider the case when there is a real intersection in the lower half-plane.
Find a stationary rotation $M_\eps \in U_\eps(M)$ such that all eigenvalues of $\Omega(M_\eps)$ are distinct. Then the rotation $M_\eps$ is unstable. Moreover, an argument similar to the one of \cite{Marshall} can be used to show that there is a heteroclinic trajectory joining $M_\eps$ with another stationary rotation $M^{-}_\eps$ such that the distance $\mathrm{dist}(M_\eps, M_\eps^-)$ is uniformly bounded from below as $\eps \to 0$. Therefore $M$ is unstable.

 \section{Proof of Theorem \ref{stabThm}: stability}\label{stabSect}
 
According to Theorem \ref{genStabThm}, to prove the stability part of Theorem \ref{stabThm} we should do the following:
\begin{enumerate}
	\item check that the pencil $\Pi$ is fine at $M$ (Section \ref{sectFine});
	\item check that the equilibrium point $M$ is strongly regular (Section \ref{sectSReg});
	\item check that the spectrum of $\Pi$ at $M$ is real (Section \ref{sectSpecReal});
	\item check that the pencil $\Pi$ is diagonalizable at $x$ (Section \ref{sectDiag});
	\item check that for each $\lambda \in \Lambda_\Pi(x)$ the $\lambda$-linearization $\diff_\lambda \Pi(x)$ is compact (Section \ref{sectComp}).
\end{enumerate}

\subsection{The pencil $\Pi$ is fine}\label{sectFine}
Let $\alpha$ be such that $\alpha \notin \Lambda(M)$ and $\alpha < \lambda_{min}^2$ where $\lambda_{min}$ is the minimal eigenvalue of $J$. Take $\eps$ such that $\alpha + \eps <  \lambda_{min}^2$, and $(\alpha - \eps, \alpha + \eps) \cap \Lambda(M) = \varnothing$. Take $U = (\alpha - \eps, \alpha + \eps)$.
Then for each $\beta \in U$, the bracket $P_\beta$ is compact semisimple. Therefore conditions 1 and 2 of Definition \ref{finePencil} are satisfied. Further, for any $\beta \in U$, the map
$$
F_{\alpha\beta} \colon (\so(n),\LieBracket_\alpha)  \to (\so(n), \LieBracket_\beta)
$$ defined by
$$
F_{\alpha\beta}(X) = (J^2 - \beta E)^{-{1}/{2}}(J^2 - \alpha E)^{{1}/{2}}X(J^2 - \alpha E)^{{1}/{2}} (J^2 - \beta E)^{-{1}/{2}}
$$
is an isomorphism of Lie algebras. Therefore, for any $f_\alpha \in \zenter(P_\alpha)$, the function
$$
f_\beta(x) = f_\alpha(F_{\alpha\beta} ^*(x))
$$
is a Casimir function of $P_\beta$, so condition 3 of Definition \ref{finePencil} is also satisfied, and the pencil is fine at every point.

\subsection{Strong regularity}\label{sectSReg}
Proposition \ref{reg1} shows that a regular stationary rotation (in the sense of Definition \ref{regDef}) is a regular equilibrium (in the sense of Definition \ref{regEq}). 
Now, prove that each regular stationary rotation is strongly regular (see Definition \ref{regEq2}).
Introduce the following subspaces:

 \begin{itemize} \item $K_0$ is generated by $\{E_{2i-1, 2i} - E_{2i, 2i-1}\}_{ i = 1, \dots, m}$; \item $K_1$ is generated by $\{E_{ij} - E_{ji}\}_{2m < i < j \leq n}$. \end{itemize}
 Then
 $$
K = K_0 \oplus  K_1
$$
as a Lie pencil, which means that $K_0$ and $K_1$ are Lie subalgebras with respect to all Lie structures $\LieBracket_\alpha$, and $[K_0, K_1]_\alpha = 0$.\par
Clearly, $K_0$ is Abelian with respect to all structures $\LieBracket_\alpha$, and $ K_1$ is isomorphic to $\so(n)$ with a Lie pencil given by
\begin{align*}
		[X,Y&]_\lambda = [X,Y]_0 - \lambda[X,Y]_\infty =X ( J_1^{2} - \lambda E) Y - Y ( J_1^{2} - \lambda E) X
\end{align*}
where $J_1 = \mathrm{diag}(\lambda_{2m+1}, \dots, \lambda_n)$.
Therefore the center of $ K_1$ with respect to any Lie structure $\LieBracket_\alpha$ is trivial unless $n - 2m = 2$.
So, $\zenter_\alpha(K) = K_0$ for all $\alpha$ if $n-2m \neq 2$, and $\zenter_\alpha(K) = K$ if $n-2m=2$. In both cases $\zenter_\alpha(K)$ does not depend on $\alpha$, so $\M$ is strongly regular.
\subsection{The spectrum $\Lambda(\M)$ is real}\label{sectSpecReal}
By Proposition \ref{specProp}, the spectrum $\Lambda(\M)$ is the set of horizontal coordinates of the intersection points on the parabolic diagram of $\M$. So, under the conditions of Theorem \ref{stabThm}, the spectrum is real.
\subsection{Diagonalizability}\label{sectDiag}
It is convinient to use the the following alternative definition of diagonalizability.
\begin{statement}\label{diagCrit}
Assume that $x$ is a regular equilibrium of a bi-Hamiltonian system, and that the spectrum $\Lambda_\Pi(x)$ is real.
Then the pencil $\Pi$ is diagonalizable at the point $x$ if and only if
\begin{align*}
	 \T_x^*M / \h(x) = \bigoplus\limits_{\lambda \in \Lambda(x)} \Ker P_{\lambda}(x)/\h(x).
	\end{align*}
	
	\end{statement}
	For the proof, see \cite{JGP, SBS}.
\begin{statement}\label{diagMan}
Let $M$ be a regular stationary rotation.
	Then the pencil is diagonalizable at $M$ if and only if any two parabolas in the parabolic diagram of $M$ intersect at two different points.
\end{statement}
\begin{proof}
	Proposition \ref{diagCrit} implies that the pencil is diagonalizable if and only if
	\begin{align}\label{diagEq}
	  \so(n) / K = \bigoplus\limits_{\lambda \in \Lambda(M)} \Ker\left( P_{\lambda}\right) / K.
	\end{align}
	Using (\ref{sonDec}), write
	\begin{align*}
 		 \so(n) / K = \bigoplus\limits_{1 \leq i < j \leq m}  V_{ij} \oplus 
	\bigoplus_{\substack{1 \leq i \leq m, \\ 2m <j \leq n}}  W_{ij}.
 	\end{align*}
	Since all the summands of this decomposition are pairwise orthogonal with respect to $P_\lambda$ (Proposition \ref{KcommonKer}), relation (\ref{diagEq}) is satisfied if and only if
	\begin{align}
 		 &  V_{ij} = \bigoplus\limits_{\lambda \in \Lambda(M)}V_{ij}(\lambda) \mbox{ for } 1 \leq i < j \leq m, \label{rel1}\\
 		   W_{ij} &= \bigoplus\limits_{\lambda \in \Lambda(M)}{ W_{ij}}(\lambda) 	\mbox{ for } 1 \leq i \leq m, 2m <j \leq n\label{rel2}
	 	\end{align}	
		where $V_{ij}(\lambda)$ and $W_{ij}(\lambda)$ are defined by (\ref{vijLambda}).\par
		Since there is a unique $\lambda = \lambda_j^2$ such that $  W_{ij} = W_{ij}(\lambda)$, equality (\ref{rel2}) is always satisfied.
Equality (\ref{rel1}) is satisfied if and only if equation (\ref{specEq}) has two distinct roots, i.e. if corresponding two parabolas are not tangent to each other, q.e.d.				
\end{proof}
\subsection{Compactness}\label{sectComp}
Show that under the conditions of Theorem \ref{stabThm}, the pencil $\diff_\lambda \Pi(M)$ is compact. \par First, consider the case $\lambda = \infty$. Then $\g_\lambda$ is the $\ad^*$ stabilizer of $M \in \so(n)^*$, so  $\g_\lambda$ is compact, and so is the pencil $\diff_\infty\Pi(M)$ (see Proposition \ref{compcomp})\par
So, let $\lambda \neq \infty$. The pencil $\diff_\lambda \Pi(M)$ is defined on $\g_\lambda$ by the cocycle $\mathcal B = P_\infty\mid_{\Ker P_\lambda}$. Prove that there exists $X \in \zenter(\Ker \mathcal B)$ such that the form
 $$\mathfrak B_X(Y,Y) \defeq P_\infty([X,Y]_\lambda, Y)$$ 
 is positive definite on $\Ker P_\lambda / \Ker \mathcal B$. \par
 By Proposition \ref{diagMan}, the pencil is diagonalizable at $M$. This implies that (see Definition \ref{defDiag})
 $$
\dim \Ker (P_\infty\mid_{\Ker P_\lambda}) = \dim \Ker P_\infty.
 $$
Since $\Ker P_\infty = K$ (Proposition \ref{reg1}),
 $$
\dim \Ker (P_\infty\mid_{\Ker P_\lambda}) = \dim K.
 $$
 On the other hand, $\Ker (P_\infty\mid_{\Ker P_\lambda}) \supset K$. So,
\begin{align}\label{kerOnKer}
 \Ker \mathcal B = \Ker (P_\infty\mid_{\Ker P_\lambda}) =  K,
\end{align}
 and the compactness condition can be reformulated as follows: there exists $X \in \zenter(K)$ such that the form $\mathcal B_X$ is positive definite on 
 	$$
		\Ker P_{\lambda} / K =  \bigoplus\limits_{1 \leq i < j \leq m} V_{ij}(\lambda) \oplus 
	\bigoplus_{\substack{1 \leq i \leq m, \\ 2m <j \leq n}}  W_{ij}(\lambda).
	$$
 At the same time, $\zenter(K) =  K_0$ or $\zenter(K) = K$ (Section \ref{sectSReg}), so $K_0 \subset \zenter(K)$, and it suffices to show that there exists $X \in K_0$ such that $\mathcal B_X$ is positive on $\Ker P_{\lambda} / K$. Represent an element $X \in K_0$ as
$$
X = \left(\begin{array}{ccccccc}0 & x_1 &  &  &  &  &    \\-x_1 & 0 &  &  &  &  &    \\ &  & \ddots &  &  &  &    \\ &  &  & 0 & x_m &  &    \\ &  &  & -x_m & 0 &  &    \\ &  &  &  &  & 0 &    \\ &  &  &  &  &  & \ddots  \end{array}\right).
$$
Denote 
$$
b_{i} \defeq a_{2i} + a_{2i-1}.
$$
 \begin{statement}\label{Pinfty}
Let $ V_{ij}(\lambda) \neq 0$,
 $$
Y=\left(\begin{array}{c}\alpha \\\beta\end{array}\right) \in V_{ij}(\lambda), \quad Z = \left(\begin{array}{c}\widetilde\alpha \\ \widetilde\beta\end{array}\right) \in V_{ij}(\lambda)
 $$ 
 where $\alpha, \beta$ are the coordinates on $V_{ij}(\lambda)$ given by (\ref{alphabeta}).
	Then
\begin{align*}
P_\infty(Y,Z) = 2m_ia_{2j-1}\left({m_i^2}{b_j}-{m_j^2}{b_i}\right)(\alpha\widetilde \beta - \widetilde \alpha \beta).
\end{align*}

 \end{statement}
 \begin{proof}
Use Proposition \ref{PLambdaRestr}.
\end{proof}

 \begin{statement}\label{BXonV}
Let
 $$
 Y = \left(\begin{array}{c}\alpha \\\beta\end{array}\right) \in  V_{ij}(\lambda) \neq 0.
 $$ 
  where $\alpha, \beta$ are the coordinates on $V_{ij}(\lambda)$ given by (\ref{alphabeta}).
	Then
\begin{align}\label{aXfla1}
\mathcal B_X(Y,Y) = - 2a_{2j-1}  \left({m_i^2}{b_j}-{m_j^2}{b_i}\right)(m_ix_i-m_jx_j)({a_{2i}\alpha^2 + a_{2i-1}\beta^2}).
\end{align}

 \end{statement}
  \begin{proof}
Use Proposition \ref{Pinfty} and Proposition \ref{eigen1}.
\end{proof}
  \begin{statement}\label{BXonW}
Let
 $$
 Y = \left(\begin{array}{c}\alpha \\\beta\end{array}\right) \in  W_{ij}(\lambda) \neq 0
 $$ 
 where $\alpha, \beta$ are the coordinates on $W_{ij}(\lambda) = W_{ij}$ given by its natural identification with $\R^2$.
	Then
\begin{align}\label{aXfla2}
\mathcal B_X(Y,Y) = -2m_ix_i(a_{2i-1}\alpha^2 + a_{2i}\beta^2).
\end{align}
\end{statement}
 \begin{proof}
Use Proposition \ref{PLambdaRestr} and Proposition \ref{eigen2}. \end{proof}

\begin{statement}
Under the conditions of Theorem \ref{stabThm}, the pencil $\diff_\lambda \Pi(M)$ is compact.
\end{statement}
\begin{proof}
It suffices to consider the case $\lambda \neq \infty$ (see beginning of the section).
 Show that there exists $X \in K_0$ such that $\mathcal B_X$ is positive on $$\Ker P_\lambda / K = \bigoplus V_{ij}(\lambda) \oplus 
	\bigoplus W_{ij}(\lambda).$$
The summands of this decomposition are pairwise orthogonal with respect to $\mathcal B_X$, so it suffices to show that there exists $X$ such that $\mathcal B_X$ is positive-definite on each of the summands. Propositions \ref{BXonV} and \ref{BXonW} show that \textit{for each summand there exists $X$ such that $\mathcal B_X$ is positive on this summand} (note that $a_{2i}$ and $a_{2i-1}$ are of the same sign; this is because all intersections are in the upper half-plane). However, it is not obvious a priori why there should exist $X$ such that $\mathcal B_X$ is positive on \textit{all summands}.
Nevertheless, such an $X$ exists and can be defined by the following magical formula
\begin{align}\label{magicFormula}
x_i \defeq -\frac{m_i}{b_i}
\end{align} 
for $ i = 1, \dots, m$. Show that $\mathcal B_X > 0$ on $ V_{ij}(\lambda)$.
Substituting (\ref{magicFormula}) into (\ref{aXfla1}), obtain
\begin{align*}
\mathcal B_X(Y,Y) = 2a_{2j-1}b_j ({a_{2i}\alpha^2 + a_{2i-1}\beta^2})b_i\left(\frac{m_i^2}{b_i}-\frac{m_j^2}{b_j}\right)^2.
\end{align*}
Since all intersections are in the upper half-plane, $a_{2j-1}$ and $b_j$ have the same sign. The same is true for $a_{2i}, a_{2i-1}$ and $b_i$. Consequently, for $Y \neq 0$, the following inequality is satisfied 
$$
a_{2j-1}b_j ({a_{2i}\alpha^2 + a_{2i-1}\beta^2})b_i > 0.
$$
Further, if $$\frac{m_i^2}{b_i}-\frac{m_j^2}{b_j} = 0,$$
then $P_\infty(Y,Y) = 0$ (see Proposition \ref{Pinfty}), and $$ V_{ij}(\lambda) \subset \Ker \left(P_\infty\mid_{\Ker P_\lambda}\right),$$ which contradicts (\ref{kerOnKer}). So, $\mathcal B_X$ is positive on $ V_{ij}(\lambda)$.\par \smallskip
Now, show that $\mathcal B_X > 0$ on $ W_{ij}(\lambda)$. Substituting (\ref{magicFormula}) into (\ref{aXfla1}), obtain
$$
\mathcal B_X(Y,Y) = 2\frac{m_i^2}{b_i}(a_{2i-1}\alpha^2 + a_{2i}\beta^2).
$$
Since $a_{2i}, a_{2i-1}$ and $b_i$ are of the same sign, and $m_i \neq 0$, the form $\mathcal B_X > 0$ on $ W_{ij}(\lambda)$.\par\smallskip

\end{proof}

\begin{remark}
	Let $\lambda > \lambda_{max}^2$ or $\lambda < \lambda_{min}^2$ where $\lambda_{min}$ and $\lambda_{max}$ are, respectively, the minimal and the maximal eigenvalues of $J$. Then the compactness of $\diff_\lambda \Pi(M)$ is natural, since the algebra $\so(n)$ with the $ \LieBracket_\lambda$ bracket is compact, and so is $\g_\lambda(M)$ which is the $\ad^*_\lambda$ stabilizer of $M$. \par However, for $\lambda \in [\lambda_{min}^2, \lambda_{max}^2]$, the algebra $\g_\lambda(M)$ is not necessarily compact. The classification of Lie algebras $\g_\lambda(M)$ up to an isomorphism is given in the appendix.
	
	\section*{Appendix:  classification of Lie algebras $\g_\lambda$}
		Let $M$ be a regular stationary rotation, $\lambda \in \CP$. By $\Sigma_\lambda$ denote the set of intersection points on the parabolic diagram of $M$ with abscissa $x = \lambda$. For $\lambda \in \R$, denote by $\Sigma_\lambda^+, \Sigma_\lambda^-$ the set of intersection points $\in \Sigma_\lambda$ which lie in the upper and lower half-plane respectively. For $z \in \Sigma_\lambda$ denote by $n_z$ the number of parabolas passing through $z$. For $z \in \Sigma_\lambda^+$ denote by $l_z, r_z$ the number of parabolas passing through $z$ such that their vertices are to the left or right from $z$ respectively. 
	Denote by $v$ the number of vertical lines on the parabolic diagram. For $\lambda \in \R$, denote by $l_\lambda$ and $r_\lambda$ the number of vertical lines to the left or to the right from the line $x = \lambda$ respectively.

\end{remark}
\begin{statement}\label{gLambda}
Let $M$ be a regular stationary rotation.  
\begin{enumerate}
\item If $\lambda \in \R$, and there is no vertical line $x = \lambda$ on the parabolic diagram of $M$, then
$$
\g_\lambda \simeq \so(l_\lambda, r_\lambda) \oplus \bigoplus_{z \in \Sigma_\lambda^+} \mathfrak u(l_z,r_z) \oplus \bigoplus_{z \in \Sigma_\lambda^-} \mathfrak{gl}(n_z, \R) \oplus \R^N.$$

\item If $\lambda = \infty$, then
$$
\g_\lambda \simeq \so(v, \R) \oplus \bigoplus_{z \in \Sigma_\lambda} \mathfrak{u}(n_z, \R) \oplus \R^N.$$

\item If $\lambda \in \Complex \setminus \R$, then
$$
\g_\lambda \simeq \so(v, \Complex) \oplus \bigoplus_{z \in \Sigma_\lambda} \mathfrak{gl}(n_z, \Complex) \oplus \Complex^N.$$
\item If $\lambda \in \R$, and there is a vertical line $x = \lambda$ on the parabolic diagram of $M$, then
\begin{align*}
\g_\lambda \simeq &\left(\so(l_\lambda, r_\lambda) \ltimes_{\rho_1} \R^{l_\lambda + r_\lambda} \right)\oplus \bigoplus_{z \in \Sigma_\lambda^+}\left( \mathfrak u(l_z,r_z) \ltimes_{\rho_2} \Complex^{l_z+r_z}\right)\oplus \\ 
&\qquad\quad\oplus \bigoplus_{z \in \Sigma_\lambda^-} \left(\mathfrak{gl}(n_z, \R) \ltimes_{\rho_3} \R^{2\vphantom{l}n_z}\right) \oplus \R^N
\end{align*}
where representations $\rho_1, \rho_2$ are standard, and
$$
\rho_3(A) = \left(\begin{array}{cc}A  & 0 \\0  &  -A^\mathrm{t}\end{array}\right).
$$
\end{enumerate}
In all cases $N$ is some number $\geq 0$.
\end{statement}

\begin{figure}

{\begin{picture}(500,70)

\put(20,-10){
\qbezier(24,80)(47,-50)(70,80)
\qbezier(64,80)(87,-50)(110,80)
\qbezier(58,80)(127,-50)(192,80)

\put(0,30){\vector(1,0){180}}
{\color{black}\multiput(67,30)(0,5){7}{\line(0,1){3}}}
\put(63,22){\small{$ \lambda$}}

}
\put(210,-10){
\qbezier(154,80)(97,-50)(40,80)
\qbezier(44,80)(67,-50)(90,80)
\qbezier(38,80)(107,-50)(172,80)
\put(0,30){\vector(1,0){180}}
{\color{black}\multiput(47,30)(0,5){7}{\line(0,1){3}}}
\put(43,22){\small{$ \lambda$}}

}
\end{picture}}
\caption{Rotations with $\g_\lambda \simeq \mathfrak{u}(1,2) \oplus \R^N$ and $\g_\lambda \simeq \mathfrak{u}(3) \oplus \R^N$ respectively.}\label{supd}
\end{figure}
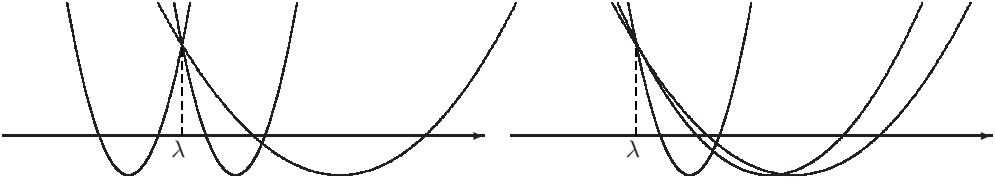

\begin{example}
Rotations with $\g_\lambda \simeq \mathfrak{u}(1,2) \oplus \R^N$ and $\g_\lambda \simeq \mathfrak{u}(3) \oplus \R^N$ respectively are depicted in Figure \ref{supd}.  Proposition \ref{spectra} can be used to check that both cases correspond to a $(1:1:1)$ resonance. 
Note that, in both cases, the bi-Hamiltonian system corresponding to the linear pencil $\diff_\lambda\Pi$ coincides with the three-wave interaction system \cite{Alber}. So, the three-wave interaction system is the bi-Hamiltonian linearization of the multidimensional rigid body at a $(1:1:1)$ resonance.

\end{example}
	\bibliographystyle{unsrt} 
\bibliography{Diss} 
\end{document}